\documentclass{article}
\usepackage{setspace}
\usepackage[colorlinks,linkcolor=black,citecolor=black,urlcolor=black,%
            bookmarks=false]{hyperref}
\usepackage[ansinew]{inputenc}
\usepackage{hyperref}
\usepackage[british]{babel}
\usepackage[reqno]{amsmath}
\usepackage{amsthm,amsfonts,amssymb}
\usepackage[mathscr]{euscript}
\usepackage{graphicx,color}
\usepackage{multicol,enumerate,fancyhdr,emptypage,xspace,subfig}
\usepackage{makeidx}
\usepackage[ruled,vlined]{algorithm2e}

\newtheorem{proposition}{Proposition}[section]
\newtheorem{cor}{Corollary}[section]

\newtheorem{definition}{Definition}[section]
\newtheorem{theorem}{Theorem}[section]
\newtheorem{remark}{Remark}[section]
\newtheorem{example}{Example}[section]

\newcommand{\Z}{{\mathbb Z}}
\newcommand{\ZZ}{\Z_{p^r}}

\newcommand{\FF}{\mathbb{F}_p}

\newcommand{\C}{{\mathcal C}}

\usepackage{amssymb}

\makeindex

\newcommand{\vu}{{{\bf u}}}
\newcommand{\vv}{{{\bf v}}}

\begin{document}

\title{An Algebraic Invariant for Free Convolutional Codes over Finite Local Rings }

\author{Mohammed El Oued}

\maketitle

\begin{abstract}

This paper investigates the algebraic structure of free convolutional codes over the finite local ring $\mathbb{Z}_{p^r}$. We introduce a new structural invariant, the Residual Structural Polynomial, denoted by $\Delta_p(\mathcal{C}) \in \mathbb{F}_p[D]$. We construct this invariant via encoders which are reduced internal degree matrix (RIDM). We formally demonstrate that $\Delta_p(\mathcal{C})$ is an intrinsic characteristic of the code, invariant under equivalent RIDMs.  A central result of this work is the establishment that $\Delta_p(\mathcal{C})$ serves as an algebraic   criterion for intrinsic catastrophicity: we prove that a free code $\mathcal{C}$ admits a non-catastrophic realization if and only if $\Delta_p(\mathcal{C})$ is a monomial of the form $D^s$. Furthermore, we establish a fundamental duality theorem, proving that $\Delta_p(\mathcal{C}) = \Delta_p(\mathcal{C}^\perp)$. This result reveals a deep structural symmetry, showing that the  "catastrophicity" of a free code is preserved under orthogonality.

\end{abstract}

\section{Introduction}

Convolutional codes over finite fields have been extensively studied since the seminal work of Forney \cite{for,for1} and Massey \cite{jl}. In the classical framework, a convolutional code $\mathcal{C}$ is a sub-module of $(\mathbb{F}_q((D)))^n$ generated by a $k \times n$ polynomial matrix $G(D)$. A key concept in this theory is catastrophicity: an encoder is called catastrophic if an input sequence of infinite Hamming weight results in a codeword of finite Hamming weight. Over a field, this property is linked to the minors of the generator matrix: an encoder is non-catastrophic if and only if the greatest common divisor (GCD) of its $k \times k$ minors is of the form $D^s$ for some $s \in \mathbb{N}$, \cite{jl}.

The transition from fields to the finite ring $\mathbb{Z}_{p^r}$ introduces profound algebraic challenges.  Notably, there exist codes where every encoder is catastrophic (see \cite{mit},\cite{f}). In such cases, catastrophicity is an intrinsic structural defect of the code itself.

The primary objective of this paper is to formalize this intrinsic behavior through a new polynomial invariant. We introduce the \textit{Residual Structural Polynomial}, denoted by $\Delta_p(\mathcal{C})$. This invariant is defined using \textit{Reduced Internal Degree Matrices}(RIDM), which provide a minimal representation of the code. 

Our main contributions are as follows. First, we prove that $\Delta_p(\mathcal{C})$ is a well-defined intrinsic characteristic of any free convolutional code over $\mathbb{Z}_{p^r}$. Second, we establish that this polynomial provides a definitive criterion for catastrophicity: a code is  non catastrophic if and only if $\Delta_p(\mathcal{C}) = D^s$ for some $s \geq 0$. Finally, we present a  duality theorem, proving that $\Delta_p(\mathcal{C}) = \Delta_p(\mathcal{C}^\perp)$. This equality reveals that the "catastrophicity" is a symmetric property shared between a code and its dual.

\section{Preliminaries results}
This section establishes the foundational algebraic concepts and notations used throughout this paper. We also present the necessary results required to address the problems in the subsequent sections. Throughout this paper, $p$ denotes a prime number and $r$ a positive integer. We consider the finite local ring $\mathbb{Z}_{p^r}$. For any element $a \in \mathbb{Z}_{p^r}$, we denote by $\overline{a}$ the natural projection (reduction modulo $p$) onto the residue field $\mathbb{F}_p$. This projection is extended to the ring of polynomials $\mathbb{Z}_{p^r}[D]$, the field of rational fractions $\mathbb{Z}_{p^r}(D)$, and the ring of Laurent series $\mathbb{Z}_{p^r}((D))$  in the usual way. \\

\subsection{Convolutional codes over a finite ring}
A large part of the existing literature on convolutional codes over rings utilizes the framework of semi-infinite Laurent series to represent code sequences \cite{fa,z,ha,ma}. We adopt this approach in our definition and focus on free codes.
\begin{definition}\rm\medskip 
\noindent A free convolutional code $\mathcal{C}$ of length $n$ and rank $k$ over $\mathbb{Z}_{p^r}$ is a free sub-module of $\mathbb{Z}_{p^r}((D))^n$ generated by the rows of a $k \times n$ polynomial matrix $G(D) \in \mathbb{Z}_{p^r}[D]^{k \times n}$ of full rank:
\begin{eqnarray*}
 \C &=& Im_{\ZZ((D))}G(D)\\
&=& \{\vv(D)\in\ZZ((D))^n/\exists \vu(D)\in\ZZ((D))^k\ :\vv(D)=\vu(D) G(D)\}.
\end{eqnarray*}
\end{definition}

The matrix $G(D)$ is called a generator matrix (or encoder) for $\mathcal{C}$.\\

Two generator matrices $G_1(D), G_2(D) \in \mathbb{Z}_{p^r}[D]^{k \times n}$ are equivalent  if they generate the same code. This is equivalent to the existence of a $k \times k$ invertible matrix $M(D)$ over $\mathbb{Z}_{p^r}(D)$ such that $G_1(D) = M(D)G_2(D)$ \cite{z}.\\

\begin{definition}\cite{z,f}\rm\medskip

A generator matrix $G(D)\in \ZZ[D]^{k\times n}$ is defined to be  non-catastrophic if any infinite weight input $\vu(D)\in\ZZ((D))^k$ cannot produce a finite weight output, i.e.,
$$\mbox{wt}[\vu(D)G(D)]<+\infty\Longleftrightarrow \mbox{wt}[\vu(D)]<+\infty,$$
where the Hamming weight wt$[\vu(D)]$ of a sequence    $\vu(D)=(\vu_1(D),\ldots,\vu_k(D))\in\ZZ((D))^k$ is defined as the sum of the weights of its components : $$ \mbox{wt}[\vu(D)]=\displaystyle\sum_{i=1}^k \mbox{wt}[\vu_i(D)]$$ where $ \mbox{wt}[\vu_i(D)]$ represents the number of nonzero coefficients of the series $\vu_i(D)$.
\end{definition}

  In contrast to convolutional codes over fields, catastrophicity is an inherent property of codes defined over $\ZZ$ \cite{mit,z}. Specifically, there exist convolutional codes over $\ZZ$
 for which every polynomial encoder is catastrophic. Such codes are termed  catastrophic codes. For instance, the code over $\Z_4$  generated by the matrix $\left(
                                                             \begin{array}{cc}
                                                               1+D & 3+D \\
                                                             \end{array}
                                                           \right)$
is a catastrophic code.

\subsection{Reduced integral degree matrices}
We introduce now a class of generator matrices who play a crucial role in this work.
\begin{definition}\rm\medskip
 We define the internal degree of a generator matrix $G(D) \in \mathbb{Z}_{p^r}[D]^{k \times n}$, denoted as $\text{intdeg}(G)$, as the maximum of the degrees of all $k \times k$ minors of $G(D)$:
\begin{equation}
\text{intdeg}(G) = \text{maximum degree of}\; G(D)'s\; k\times k \; \text{minors}\;\cite{rj}.
\end{equation}
\end{definition}
 
\begin{definition}\rm\medskip
A generator matrix $G(D)$ is called a {\bf Reduced Internal Degree Matrix} (RIDM) if its internal degree is minimal among all polynomial generator matrices equivalent to $G(D)$.
\end{definition}
The existence of a RIDM for any free code is guaranteed by the fact that the internal degree  is a non-negative integer. We show next how construct a RIDM from any full row rank polynomial matrix. For this, we need the following classical properties for polynomial and polynomial matrices over the ring $\ZZ$.
\begin{proposition}\label{p1}\rm\medskip
Let $G(D)\in\ZZ^{k\times n}[D]$. If an irreducible  polynomial $P$ divides all $k\times k$ minors of $G$, then there exists a unimodular matrix $T(D)$ such that $P$ divides  one row of $T(D)G(D)$.
\end{proposition}
We recall that a matrix $T(D)\in\ZZ[D]^{k\times k}$ is  unimodular  if $\det M$ is a unit in $\ZZ[D]$.  Note that, $T(D)$ is unimodular if and only if its projection $\overline{M}(D)$ is unimodular over $\Z_p[D].$
\begin{theorem}{[Theorem XIII.6\cite{mc}}]\label{tech}\normalfont

 Let $P$ be a regular polynomial in $\Z_{p^r}[D]$. Then there is a monic polynomial $P_1$ with $\overline{P}=\overline{P_1}$ and a unit polynomial $P_2$ in $\Z_{p^r}[D]$ such that $P=P_1P_2.$

\end{theorem}
A characterization of a RIDM is given in the next proposition.
\begin{proposition}\label{PP}\rm\medskip
A full-rank polynomial matrix $G(D) \in \mathbb{Z}_{p^r}[D]^{k \times n}$ is a RIDM if and only if the GCD of its $k \times k$ minors is equal to 1.
\end{proposition}
\begin{proof}
Let $G(D)$ be a full row rank matrix in $\mathbb{Z}_{p^r}[D]^{k \times n}$.

 Suppose that the GCD of the $k \times k$ minors of $G$ is equal to a non-constant polynomial $P$. Since $G$ is a full row rank matrix, $P$ is a regular polynomial. By the Theorem \ref{tech}, $P$ admits an irreducible divisor $Q$ (where the lowest degree coefficient is a unit). By Proposition~\ref{p1}, there exists a unimodular matrix $T(D)$ such that one row of $T(D)G(D)$ is divisible by $Q$. By dividing this specific row by $Q$, the resulting matrix remains equivalent to $G(D)$ but has a strictly lower internal degree. This contradicts the minimality of of the internal degree of $G(D)$. Thus, the GCD of the $k \times k$ minors of any RIDM must be $1$.

 conversely, suppose that the GCD of the $k \times k$ minors of $G$ is $1$. Let $G'$ be a RIDM equivalent to $G$. It follows from the first part of this proof that the GCD of its $k \times k$ minors is also $1$. Since both matrices generate the same free module, there exists a non singular  matrix $M(D) \in \mathbb{Z}_{p^r}(D)^{k \times k}$ such that $G(D) = M(D)G'(D)$.  The $k \times k$ minors of $G(D)$ are equal to the $k \times k$ minors of $G'(D)$ multiplied by $\det M(D)$. Since the GCD of the minors of $G'(D)$ is 1, $\det M(D)$ must be the GCD of the minors of $G(D)$, which is 1. Therefore, $\det M(D)$ is a unit in $\mathbb{Z}_{p^r}$, which implies that the internal degree of $G(D)$ is equal to the internal degree of $G'(D)$. Thus, $G(D)$ is a RIDM.
\end{proof}
We now present an algorithm to transform any generator matrix $G(D)$ into an equivalent RIDM. This process works by reducing the GCD of the $k \times k$ minors within the ring $\mathbb{Z}_{p^r}[D]$ until it becomes 1.

\begin{center}
\textbf{Algorithm: Construction of an RIDM encoder}
\end{center}

\medskip
\noindent Let $G(D) \in \mathbb{Z}_{p^r}[D]^{k \times n}$ be a full row rank matrix. The following procedure constructs an equivalent Reduced Internal Degree Matrix (RIDM):

\begin{enumerate}
    \item[\textbf{Step 1.}] Compute $P(D)$ the greatest common divisor of all $k \times k$ minors of $G(D)$. 
    \item[\textbf{Step 2.}] If $P(D)$ is a unit (i.e., $GCD = 1$), then $G(D)$ is a RIDM. The algorithm terminates.
    \item[\textbf{Step 3.}] If $P(D)\neq1$, then $P$ is a non constant, non-unit polynomial. Let $Q(D)$ be an irreducible divisor of $P(D)$. 
    \item[\textbf{Step 4.}] Apply a unimodular transformation $T(D)$ to $G(D)$ such that one row of the product $T(D)G(D)$ is divisible by $Q(D)$. Let this row be $g_i(D)$.
    \item[\textbf{Step 5.}] Replace $g_i(D)$ by $\frac{1}{Q(D)}g_i(D)$ to obtain a new generator matrix $G'(D)$.
    \item[\textbf{Step 6.}] Repeat Steps 1 to 5 for $G'(D)$ until the GCD of the $k \times k$ minors is 1.
\end{enumerate}
The algorithm terminates when the GCD of the $k \times k$ minors equals 1, yielding a RIDM equivalent to the initial matrix.

\begin{remark}\rm\medskip
A fundamental distinction between convolutional codes over fields and those over rings lies in the existence of basic encoders. While every code over a field admits a basic generator matrix, a code over $\mathbb{Z}_{p^r}$ does not necessarily possess one. However, as demonstrated by the last Algorithm, such a code always admits a RIDM encoder. It is worth noting that when the RIDM condition is translated to the case of fields, it coincides with the standard definition of a basic matrix \cite{rj}. Over $\mathbb{Z}_{p^r}$, however, the two concepts diverge: a basic matrix must admit a polynomial left inverse, a property not guaranteed for all RIDMs.
\end{remark}
\section{The Residual Structural Polynomial}

In this section, we introduce the \textbf{Residual Structural Polynomial} as a fundamental invariant of a free convolutional code over $\ZZ$. We establish its construction from a polynomial generator matrix and  prove its independence from the choice of the Reduced Internal Degree Matrix (RIDM).

\subsection{ Definition and Invariance Theorem} 
We first define the polynomial for a specific matrix before showing it is an intrinsic property of the code.

\begin{definition}\rm\medskip
Let $G(D)$ be a RIDM over $\mathbb{Z}_{p^r}$. We define the \textbf{Residual Structural Polynomial} $\Delta_p(G)$ as the GCD of the $k \times k$ minors of the projected matrix $\overline{G}(D) \in \mathbb{F}_p[D]^{k \times n}$.

\end{definition}

\begin{remark}\rm\medskip
If two generator matrices $G_1(D)$ and $G_2(D)$ are related by a non-singular matrix  $M \in GL_k(\ZZ[D])$, we have:
\begin{equation}
\Delta_p(G_1) = \lambda \det(\overline{M}) \Delta_p(G_2)
\end{equation}
where $\lambda \in \FF \setminus \{0\}$ is a scalar such that the resulting polynomial is monic. In particular if $M$ is unimodular, we have $\Delta_p(G_1)=\Delta_p(G_2)$.
\end{remark}

The most critical property of this polynomial is its invariance. As shown in the next theorem, $\Delta_p(G)$ remains constant across all equivalent RIDM encoders. This confirms that the polynomial is an intrinsic property of the code $\mathcal{C}$, allowing us to denote it as $\Delta_p(\mathcal{C})$.

\begin{theorem}\label{imp}\normalfont
Let $G_1$ and $G_2$ be two equivalent reduced internal degree matrices, then $\Delta_p(G_1)=\Delta_p(G_2)$.
\end{theorem}
\begin{proof} Let $G_1(D)$ and $G_2(D)$ be two $k \times n$ matrices satisfying the conditions of the theorem. Then there is a $k \times k$ nonsingular matrix $M(D)$ over $\ZZ(D)$ such that $G_1(D) = M(D)G_2(D)$. This implies that the $k\times k$ minors of $G_1$ are equal to the $k\times k$ minors of $G_2$, each multiplied by $\det M$. By Proposition \ref{PP}, the GCD of the $k\times k$ minors for any RIDM is 1. This shows  that $\det M$ must be  a unit  in $\ZZ[D]$. This proves that $\Delta_p(G_1)=\Delta_p(G_2)$.

\end{proof}
Consequently, the following definition is well-established.
\begin{definition}\normalfont
Let $\C$ be a free convolutional code over $\ZZ$. if $G(D)$  is any reduced internal degree encoder for $\C$. The \textbf{Residual Structural Polynomial}  of the code is:  $$\Delta_p(\C)=\Delta_p(G).$$
\end{definition}

\subsection{Characterization of Intrinsic Catastrophicity} 
The polynomial $\Delta_p(\mathcal{C})$ provides a direct algebraic diagnosis of the catastrophicity of the code. Unlike codes over fields, catastrophicity over $\mathbb{Z}_{p^r}$ can be an unavoidable structural property.\\

A key characteristic of catastrophic polynomial generator matrices over $\ZZ$ lies in their projection over the residue field:
\begin{theorem}\cite{z}\label{j12}\normalfont\\
A polynomial matrix $G(D)$ over $\mathbb{Z}_{p^r}$ is catastrophic if and only if its projection $\overline{G}(D)$ is catastrophic over $\mathbb{F}_p$.
\end{theorem}
While the previous theorem applies to a specific matrix, the invariance of $\Delta_p(\mathcal{C})$ allows us to characterize the code itself.
\begin{theorem}\rm\medskip 
A free convolutional code $\mathcal{C}$ over $\mathbb{Z}_{p^r}$ is  non-catastrophic if and only if:
\begin{equation}
    \Delta_p(\mathcal{C}) = D^s, \quad \text{for some } s \geq 0.
\end{equation}
 
\end{theorem}
If $\Delta_p(\mathcal{C})$ contains any irreducible factor $Q(D) \in \mathbb{F}_p[D]$ with $Q(D) \neq D$, then every polynomial generator matrix of $\mathcal{C}$ is catastrophic. This identifies codes that are  unsuitable for practical applications due to infinite error propagation.
\begin{remark}\rm\medskip For convolutional codes over a field $\mathbb{F}_p$, every RIDM is a basic matrix, meaning its $k \times k$ minors are relatively prime \cite{rj}. Thus, $\Delta_p(\mathcal{C}) = 1$ always. The fact that $\Delta_p(\mathcal{C})$ can be different from $1$ (or $D^s$) over $\mathbb{Z}_{p^r}$ highlights the profound structural differences introduced by rings.
\end{remark}

\section{Duality and Catastrophicity Invariants}

Let $\mathcal{C}$ be a free convolutional code over $\mathbb{Z}_{p^r}$ with rank $k$ and length $n$. The dual code $\mathcal{C}^\perp$ is defined as:
\begin{equation}\label{e}
    \mathcal{C}^\perp = \{ \mathbf{v}(D) \in \mathbb{Z}_{p^r}^n[D] \mid \mathbf{v}(D) \cdot \mathbf{w}(D)^T = 0, \forall \mathbf{w}(D) \in \mathcal{C} \}.
\end{equation}

Since $\mathcal{C}$ is free, it admits a generator matrix $G(D) \in \mathbb{Z}_{p^r}^{k \times n}[D]$. Correspondingly, its dual $\mathcal{C}^\perp$ is a free code of rank $n-k$, admitting a generator matrix $H(D) \in \mathbb{Z}_{p^r}^{(n-k) \times n}[D]$, referred to as the parity-check matrix of $\mathcal{C}$ \cite{en}. These matrices satisfy the orthogonality relation:
\begin{equation}\label{e1}
    G(D) \cdot H(D)^T = \mathbf{0}_{k \times (n-k)}.
\end{equation}

\begin{theorem}\rm\medskip
Let $\mathcal{C}$ be a free convolutional code over $\mathbb{Z}_{p^r}$ and $\mathcal{C}^\perp$ be its dual code. Then:
\begin{equation}
    \Delta_p(\mathcal{C}) = \Delta_p(\mathcal{C}^\perp)
\end{equation}
\end{theorem}

\begin{proof}
Let $G(D)$ and $H(D)$ be RIDM generator matrices for $\mathcal{C}$ and $\mathcal{C}^\perp$, respectively. By the orthogonality property (\ref{e1}), the projection of these matrices onto the residue field satisfies $\overline{G}(D) \cdot \overline{H}(D)^T = 0$ over $\mathbb{F}_p(D)$.

According to the Theorem on Minors of the Inverse, as discussed by Gantmacher \cite{ga} and McEliece \cite{rj}, there exists a proportional relationship between the $k$-order minors of $\overline{G}$ and the $(n-k)$-order minors of $\overline{H}$. Specifically, each minor $\overline{m}_i$ of order $k$ from $\overline{G}$ is equal, up to  a scalar factor, to the complementary $(n-k)$-order minor of $\overline{H}$.
Consequently, we have $\Delta_p(G)=\Delta_p(H)$, and  thus $\Delta_p(\mathcal{C}) = \Delta_p(\mathcal{C}^\perp)$.
\end{proof}
\begin{example}\rm\medskip Consider the free convolutional code over $\Z_4$ with encoder 
$$G(D)=\begin{pmatrix}
1+D & 1+D & 2\\
2+D & 1 & 1+D 
\end{pmatrix}
.$$
As established previously, $G(D)$ is a RIDM for $\mathcal{C}$ with $\Delta_p(\mathcal{C}) = 1+D^2$, indicating that $\mathcal{C}$ is catastrophic. A parity-check matrix for $\mathcal{C}$ is given by:
\begin{equation*}
H(D) = \begin{pmatrix} 3+2D+D^2 & 3+3D^2 & 3+2D+3D^2 \end{pmatrix},
\end{equation*}
which is also a RIDM encoder for $\mathcal{C}^\perp$. Reducing modulo $p$, we obtain $$\overline{H}(D)=\begin{pmatrix} 1+D^2 & 1+D^2 & 1+D^2\end{pmatrix}.$$
 The GCD of the $1 \times 1$ minors is $\Delta_p(\mathcal{C}^\perp) = 1+D^2$, confirming that $$\Delta_p(\mathcal{C}^\perp) = \Delta_p(\mathcal{C}).$$
\end{example}

\begin{cor}\rm\medskip
 A code $\mathcal{C}$ over $\mathbb{Z}_{p^r}$ is catastrophic if and only if its dual $\mathcal{C}^\perp$ is catastrophic.
 \end{cor}
 
 This result highlights that catastrophicity is a structural property for the code that remains invariant under duality

\section{conclusion}
In this paper, we introduced a new structural invariant for free convolutional codes over the finite  ring $\mathbb{Z}_{p^r}$, referred to as   the Residual Structural Polynomial $\Delta_p(C)$. By utilizing the framework of Reduced Internal Degree Matrices (RIDM), we formally proved that this polynomial is an intrinsic characteristic of the code, and does not depend on the choice of a RIDM as  encoder. We established that $\Delta_p(\mathcal{C})$ provides a necessary and sufficient algebraic condition for catastrophicity: a free code over $\mathbb{Z}_{p^r}$ admits a non-catastrophic realization if and only if its residual polynomial is of the form $D^s$ for some $s \geq 0$. Furthermore, our Duality Theorem proves that $\Delta_p(C) = \Delta_p(C^\perp)$, revealing a fundamental symmetry where the catastrophicity  is  preserved under orthogonality.\\
 Future research will explore the extension of this invariant to non-free codes and investigate the explicit relationship between the irreducible factors of $\Delta_p(C)$ and the specific structure properties of catastrophic sequences.
 
\section*{Declaration of Generative AI and AI-assisted technologies in the writing process}

The author used Google's Gemini tool in preparing this work to improve the quality of the English language, writing style, and the structure of the abstract and introduction.

After using this tool, the author reviewed and corrected the content as needed and assumes full responsibility for the scientific content of the published article.

\section*{Declaration of competing interest}
The author declares that there is no conflict of interest between the funders or individuals who may influence the current work in this article.

\end{document}